\documentclass{article}

\usepackage{amsfonts}
\usepackage{amsmath}
\usepackage{amssymb}
\usepackage{amsthm}

\newtheorem{thm}{Theorem}[section]

\newtheorem{lem}[thm]{Lemma}
\newtheorem{prop}[thm]{Proposition}

\theoremstyle{definition}

\newtheorem{rem}[thm]{Remark}
\newtheorem{ex}{Example}
\newtheorem{protocol}{Protocol}
\newtheorem{problem}{Problem}

\begin{document}

\title{Group key management based on semigroup actions \thanks{The Research was supported in part by the Swiss National Science
Foundation under grant No.\@ 149716. First author is partially
supported by Ministerio de Educacion, Cultura y Deporte grant
``Salvador de Madariaga'' PRX14/00121, Ministerio de Economia y
Competitividad grant MTM2014-54439 and Junta de Andalucia (FQM0211).
The last author is supported by Armasuisse.
} } 

\author{
J.A.~L\'opez-Ramos\footnote{University of Almeria}, J.~Rosenthal\footnote{University of Zurich}, D.~Schipani\footnotemark[3], R.~Schnyder\footnotemark[3]
}

\maketitle

\begin{abstract}
In this work we provide a suite of protocols for group key
management based on general semigroup actions. Construction of the key
is made in a distributed and collaborative way. Examples are provided that may
in some cases enhance the security level and communication overheads of
previous existing protocols. Security
against passive attacks is considered and depends on the hardness of the semigroup action problem in any particular scenario.
\end{abstract}


\section{Introduction}

Traditional cryptographic tools for key exchange may not be
useful when the communication process is carried out in a group of
nodes or users. There exist several approaches for group key
management, which may be divided into three main classes
\cite{rafaeli}:
\begin{itemize}
    \item \emph{centralized} protocols, where a single entity is in
        charge of controlling the whole group, minimizing storage
        requirements, computational power on both the client and server
        side and communication overheads,
    \item \emph{decentralized}, where a large group is divided into
        subgroups in order to avoid concentrating the workload in a single
        point,
    \item \emph{distributed}, where key generation is carried out in a
        distributed and collaborative way.
\end{itemize}
This last class of approaches has become particularly
important since the emergence of ad hoc networks, where a set of
nodes,
possibly consisting of light and mobile devices, create, operate and
manage a network, which is therefore solely dependent on the
cooperative and trusting nature of the nodes.
Moreover the limited capacity of the involved devices imposes both
key storage and computational requirements. Such a network is commonly
created to meet an immediate demand and specific goal, and nodes are
continuously joining or leaving it. Thus, group key management
based on distributed and collaborative schemes has proved to be of
great interest (cf.\@ for instance \cite{vandenmerwe} and its
references).

One of the most cited approaches in the distributed
setting is due to Steiner et al.\@ in \cite{steiner1} and
\cite{steiner2}. In these works the authors provide two different
group key management schemes that extend the traditional
Diffie-Hellman key exchange~\cite{diffie} and feature very efficient
rekeying procedures.

In \cite{maze}, the authors generalize the aforementioned classical
Diffie-Hellman key exchange to arbitrary group actions:

\begin{protocol}[Semigroup Diffie-Hellman Key Exchange]
Let
$S$ be a finite set, $G$ an abelian semigroup, and $\Phi: G\times S \to S$ a
$G$-action on $S$. The semigroup Diffie-Hellman key exchange in $(G,
S, \Phi)$ is the following protocol:

\begin{enumerate}

\item Alice and Bob publicly agree on an element $s\in S$.

\item Alice chooses $a\in G$ and computes $\Phi (a,s)$. Alice's private key is
$a$, her public key is $\Phi (a,s)$.

\item Bob chooses $b\in G$ and computes $\Phi (b,s)$. Bob's private key is $b$, his
public key is $\Phi (b,s)$.

\item Their common secret key is then $$\Phi (a,\Phi (b,s)) = \Phi( ab,s) = \Phi (ba ,s) =
\Phi (b,\Phi (a,s)).$$

\end{enumerate}
\end{protocol}

In the original Diffie-Hellman proposal, if an adversary is able
to solve the so-called Discrete Logarithm Problem (DLP), then she is
able to break the Diffie-Hellman key exchange. In this setting we can
analogously consider the following more general problem:

\begin{problem}[Semigroup Action Problem, SAP]
Given a
semigroup $G$ acting on a set $S$ and elements $x,y\in S$, find
$g\in G$ such that $\Phi (g,x) = y$.
\end{problem}

It is clear that if an adversary, Eve, finds a $g \in G$
such that $\Phi (g, s) = \Phi (a,s)$, then she can find the shared
secret by computing $\Phi (g, \Phi (b,s)) = \Phi (gb, s) = \Phi (bg,
s) = \Phi (b, \Phi (a, s))$.

We can say that the security of the preceding protocol is equivalent
to the following problem.

\begin{problem}[Diffie-Hellman Semigroup Action Problem, DHSAP]
Given a finite abelian semigroup $G$ acting on a finite set
$S$ and elements $x, y, z \in S$ with $y = \Phi (g , x)$ and $z =
\Phi (h , x)$ for some $g, h \in G$, find $\Phi (gh, x)$.
\end{problem}

Although, as noted above, solving the SAP implies solving the
DHSAP, we do not know if both problems are (in general)
equivalent, just like in the traditional setting of Diffie-Hellman, where however some equivalence results for particular scenarios are known \cite{maurer}.

\medskip

Motivated by the above, our idea is now to define extensions of the
semigroup Diffie-Hellman key exchange protocol to $n$ users,
by first generalizing those introduced in \cite{steiner1} and \cite{steiner2},
and then considering other settings, which can feature more favorable
characteristics compared to the original protocol.
Since the capability of devices is often limited and authentication
processes may be difficult to implement in a distributed network, we
focus our attention on confidentiality under passive attacks. As in
\cite{maze}, some non-standard settings are introduced as more general
examples, although the hardness of the SAP there may not be proven
yet, so the security of the protocols in those cases is conditional on
that.

The structure of the paper is as follows. In Section 2 we consider a
suite of three protocols for group key management based on one-sided
actions. While these naturally extend the results of \cite{steiner1} and
\cite{steiner2}, we consider different settings for a general
semigroup action.
Section 3 considers the security of the preceding protocols against passive
attacks, including forward and backward secrecy.
Finally, in Section 4, we introduce two protocols based on linear
actions, i.e.\@ semigroup actions on other groups satisfying a certain
distributivity property. We give two different group key protocols
in this setting, one of which runs very efficiently in only two
rounds, independently of the number of members in the communicating
group.

\medskip

Throughout this paper we will consider a group of $n$ users,
${\cal U}_1, \dots , {\cal U}_n$, who would like to share a secret element of a
finite set $S$, and $G$ will denote a finite abelian semigroup acting
on $S$.

\section{Group key communication based on one-sided actions}

In this section we consider three different extensions of the
semigroup Diffie-Hellman key exchange with different computing
requirements and communication overheads, but with possible
applications in different cases. They are natural extensions of  \cite{steiner1} and
\cite{steiner2}. For completeness we report proofs in appendix to show soundness of the schemes.

\subsection{A sequential key agreement}\label{Sec11}

The first approach to extend the key exchange protocol consists of a
sequence of messages, built using pieces of private information, along
a chain of users and an analogous second sequence of messages in
the opposite way. Therefore every user will send and receive two
messages except for the user that initiates the communication and the
last user receiving the sequence of messages.

The protocol is defined by the following steps.

\begin{protocol}[GSAP-1]
Users agree on an element $s$ in a finite set $S$, a finite abelian
semigroup $G$, and a $G$-action on $S$ given by $\Phi$. For every
$i=1, \dots ,n$, the user ${\cal U}_i$ holds a private element
$g_i\in G$.

\begin {enumerate}
\item For $i = 1,\ldots, n - 1$, user ${\cal U}_i$ sends to user ${\cal U}_{i + 1}$ the message
$$\{ C_1,\dots ,C_i\} = \Big \{ \ \Phi (g_1 , s), \ \Phi (g_2g_1 ,
s),\dots , \Phi \Big ( \prod_{j=1}^{i}g_j  , s \Big)  \ \Big \}.$$

\item User ${\cal U}_n$ computes \ \ $\Phi (g_n , C_{n-1} )$.

\item For $k = n,\dots,2$, user ${\cal U}_k$ sends to user ${\cal U}_{k-1}$
the message $\big \{ f_1^k,\dots,f_{k-1}^k \big \}$, where $f_j^k = \Phi (g_k ,
f_j^{k+1} )$ for $2\leq k \leq n-1$ and  $f_j^n = \Phi (g_n ,
C_{j-1} )$, $j = 1,\dots ,n-1$, with $C_0=s$.

\item User ${\cal U}_k$ computes $\Phi (g_k , f_k^{k+1})$.
\end{enumerate}
\end{protocol}

\subsection{A key agreement in broadcast}\label{Sec12}

The following protocol presents a lower communication overhead than
GSAP-1. The idea is again to get a first sequence of messages
from user ${\cal U}_1$ to user ${\cal U}_n$, but now ${\cal U}_n$
will broadcast a message that allows the rest of the users to recover
the common key.

\begin{protocol}[GSAP-2]
Users agree on an element $s$ in a finite set $S$, a finite abelian
semigroup $G$, and a $G$-action $\Phi$ on $S$. For every $i=1,
\dots ,n$, the user ${\cal U}_i$ holds a private element $g_i\in
G$.

\begin {enumerate}
\item For $i =1, \dots ,n-1$, user ${\cal U}_i$ sends to user ${\cal
U}_{i+1}$ the message $$ \big \{ C_{i-1}^{i-1}, C_1^i,\ldots,C_i^i \big
\},$$
\noindent where \ $C_0^0 = s$, \ $C_1^1 = \Phi (g_1, s)$, and for $i
\geq 2$, $C_1^i = \Phi (g_i, C_{i-2}^{i-2})$, \ $C_j^i = \Phi (g_i,
C_{j-1}^{i-1})$ (with $j = 2, \dots ,i$).
\item User ${\cal U}_n$ computes \ \ $\Phi (g_n, C_{n-1}^{n-1})$.
\item User ${\cal U}_n$ broadcasts $ \big \{ f_1^n,\ldots,f_{n-1}^n,
f_n^n \big \}$,
where $f_i^n = \Phi (g_n, C_{n-1-i}^{n-1})$ for $i = 1, \dots ,n-2$,
$f_{n-1}^n = \Phi (g_n, C_{n-2}^{n-2})$ and $f_n^n=C_{n-1}^{n-1}$.
\item User ${\cal U}_i$ computes \ \ $\Phi (g_i , f_i^n)$.
\end {enumerate}
\end{protocol}

\begin{rem}
It can be observed that the element $f_n^n$ contained in the
broadcast message in step 3 of Protocol GSAP-2, is not needed by
any of the users ${\cal U}_i$, $i=1, \ldots ,n-1$ to recover the
shared key. However, the distribution of this value is required for
future rekeying operations, as we will later show.
\end{rem}

\subsection{Examples}\label{exgsap2}

a) The two previous protocols are extensions of those introduced in
\cite{steiner1} and \cite{steiner2} for the action of the
multiplicative semigroup $\mathbb{N}^*$ on a cyclic group $S$ of
order $q$ generated by $g$, given by $\Phi(y,g^x)=(g^x)^y$.
It was pointed out that the first protocol presents excessive
communication overheads mainly due to both the number of rounds and
messages to be sent. Because of this, only the second one, referred to as
IKA.1 in \cite{steiner2}, was recommended. However, the first protocol
could be interesting on its own when applied to a sensor network whose
communications need to be secure and where it should be assessed
whether every node is working properly. After user ${\cal U}_n$
receives the message in step 1, the absence of any of the
messages (excepting the last one) in the descending chain of rounds
would alert that the corresponding sender node is not working or the
communication was interrupted.
\medskip

\noindent b) In particular, consider a finite field $GF(q)$ and an
element $g$ of prime order.  The semigroup $\mathbb{N}^*$ acts on
the subgroup $\langle g \rangle \subset GF(q)^*$ by
$\Phi(y,g^x)=(g^x)^y$ for $x,y \in \mathbb{N}^*$.

\medskip

\noindent c) Let $\varepsilon$ be the set of points in an elliptic
curve. Then the action $\Phi : \mathbb{N}^* \times \varepsilon
\rightarrow \varepsilon$ given by $\Phi (n,P)=nP$ for every $n\in
\mathbb{N}^*$ and every $P\in \varepsilon$ provides the
corresponding versions of the preceding protocols for elliptic
curves. In \cite{quiuna} an implementation of the second protocol
can be found.

\medskip

\noindent d) In \cite[Example 5.13]{maze} the authors illustrate the
use of a semiring of order 6 to construct an example of a practical
SAP. This was later cryptanalyzed in \cite{steinwandt} due not to a
general attack, but rather to the structure of this ring. However, we
can use the semiring of order 20 given in \cite[Example 5.8]{maze}
to analogously define another SAP and its cryptanalysis is still an
open question. This shows an example where SAP does not coincide
with a traditional DLP on a semigroup and it is applicable to both
preceding protocols.

\medskip

\noindent e) In \cite[Protocol 80]{gnilke} the author defines a key
exchange protocol whose security is based on the SAP derived from
the following semigroup action: let $S$ be a semiring, $T$ a
finitely generated additive subsemigroup of $S$ and let
$\rm{End}_{+}(T)$ be its (additive) endomorphisms semigroup. Then
the semigroup action that defines the security of this protocol is
given by $\Phi : (S,T^{op})\times \rm{End}_{+}(T)\rightarrow
\rm{End}_{+}(T)$, $((s,t),f)\mapsto (x\mapsto s*f(x)*t)$.

\begin{rem}
Many examples of semigroup actions suitable to defining a
Diffie-Hellman type key exchange protocol can be found in
\cite{mazeth}. The corresponding SAP is shown to be computationally
equivalent to a DLP for some of them.
\end{rem}


\subsection{A key agreement given by a group action}\label{Sec13}

The existence of inverses in the semigroup $G$
acting on the set $S$ can provide a way to agree on a common key
with reduced communication overheads. Moreover, computations can be made more equally distributed among
the users. We remark that in the protocols given in the two
previous sections, these requirements are higher the further away the user is
from the one that initialized the protocol. 

Thus we assume that $G$ is a group. The protocol is given by the
following steps.

\begin{protocol}[GSAP-3]
Users agree on an element $C_0 = s$ in a finite set $S$, a finite abelian
group $G$, and a $G$-action $\Phi$ on $S$. For every $i=1,
\dots ,n$, the user ${\cal U}_i$ holds a private element $g_i\in
G$.

\begin{enumerate}
\item For $i=1, \dots ,n-2$, user ${\cal U}_i$ sends to user  ${\cal U}_{i+1}$ the message $C_i=\Phi (g_i, C_{i-1})$.

\item User ${\cal U}_{n-1}$ computes $C_{n-1}=\Phi (g_{n-1}, C_{n-2})$
and broadcasts it to the other users
$\{ {\cal U}_1 ,\dots ,{\cal U}_{n-2}, {\cal U}_n \}$.

\item User ${\cal U}_n$ computes the element $\Phi (g_n, C_{n-1})$.

\item  For $i=1, \dots ,n-1$, user ${\cal U}_i$ computes $D_i=\Phi (g_i^{-1} , C_{n-1})$ and sends it to user ${\cal U}_n$.

\item For $i=1, \dots ,n-1$, user ${\cal U}_n$ computes $\Phi (g_n, D_i)$
and sends to users $\{ {\cal U}_1 ,\dots ,{\cal U}_{n-2}, {\cal
U}_{n-1} \}$ the set of values $\{ \Phi (g_n, D_1), \dots , \Phi
(g_n, D_{n-1}), C_{n-1}  \}$.

\item For $i=1, \dots ,n-1$, user ${\cal U}_i$ computes $\Phi (g_i, \Phi
(g_n,D_i))$.
\end{enumerate}
\end{protocol}

After protocol GSAP-3, the users ${\cal U}_1, \dots ,{\cal U}_n$
share a common key given by $\Phi \Big( \displaystyle{\prod_{i=1}^n}
g_i , s\Big)$.
This follows easily from the commutativity of
$G$ and the fact that for every $g_i,g_j\in G$, $i,j=1, \dots ,n$ and
$s\in S$, we get that $\Phi (g_ig_j,s)= \Phi (g_i,\Phi (g_j,s))$.

\begin{rem}
As in Protocol GSAP-2, we also point out that the element
$C_{n-1}$, which is broadcast by ${\cal U}_n$ in step 5 of
Protocol GSAP-3, is needed only for future rekeying purposes.
\end{rem}

\begin{rem}  Using the action $\Phi (y, g^x) = (g^x)^y$ for
$x, y \in \mathbb{Z}_q^*$, with $g$ a generator of a cyclic group $S$
of order $q$, we get the third protocol introduced
in \cite{steiner1} and \cite{steiner2} and referred to as IKA.2 in
CLIQUES \cite{steiner2}. In this case, user ${\cal U}_i$ sends to
user ${\cal U}_n$ the message $g^{\prod_{j=1,j\not=i}^{n-1} x_j}$,
which is computed with the element $x_i^{-1} \ \mbox{mod} \ q$,
given that the $x_i$'s are selected either to be coprime with $q$ or,
as the authors suggest, $q$ is chosen to be a prime.

An elliptic curve version is clearly also feasible. An implementation in this sense can be found in \cite{quiuna}.

%
 \end{rem}


\section{Security of the key agreements and rekeying operations}\label{Secsec}

In \cite{maze} it was pointed out that if an adversary is able to
solve the SAP, then she will be able to break the
two party Diffie-Hellman key exchange, i.e.\@ solve the DHSAP. It is
easy to observe that
being able to solve the DHSAP allows getting the shared key
in all the protocols proposed above.

\begin{prop}\label{securityGSAP}
If an adversary is able to solve the DHSAP, then she can get
the shared key in GSAP-1, GSAP-2 and GSAP-3.
\end{prop}

\begin{proof}
This follows from the fact that the adversary
can access the pair of values
\begin{itemize}
    \item $(C_1, f_1^2) = \big( \Phi\big(g_1 , s \big), \Phi\big(
        \prod_{i=2}^n g_i, s \big) \big)$ in GSAP-1;
    \item $(C_1^1, f^n_1)=\big(\Phi\big(g_1 , s \big),
        \Phi\big(\prod_{i=2}^n g_i , s \big) \big)$ in GSAP-2;
    \item $\big(C_1, \Phi\big(g_n, D_1 \big)\big)= \big( \Phi\big(g_1 ,
        s \big), \Phi\big(\prod_{i=2 }^n g_j , s \big) \big)$ in GSAP-3.
\end{itemize}
\end{proof}

The preceding result shows, as could be expected, that the
multiparty key exchange protocols do not enhance the security that
the corresponding two-party protocol offers. However, as in
\cite{steiner1} and \cite{steiner2}, it is possible to show that
increasing the number of messages does not produce any information
leakage whenever the corresponding key exchange based on the SAP for
two communicating parties is secure. Here we are referring to security
against passive attacks; a totally different picture would arise if we
assume that the attacker can control communications from and to one or
more particular users, see e.g.\@ \cite{schnyder}.

\medskip

Let $X=\{ g_1, \dots ,g_n \}$ be a set of elements of the semigroup
$G$, $s$ an element of a set $S$ and $\Phi$ a $G$-action on $S$.
Let us define the (ordered) set of elements of $S$

$$V_{\Phi}^G(s,n,X)=\Big \{ \Phi
\Big( \prod_{j=i_1}^{i_m}g_j,s\Big) : \ \{ i_1,\dots ,i_m\} \subsetneq
\{ 1, \dots ,n\} \Big\}$$

\noindent and the value $K_{\Phi}^G(s,n,X)=\Phi \Big(
\prod_{j=1}^{n}g_j,s\Big)\in S$.

We point out that the messages that any adversary observes in any of
the protocols is a subset of $V_{\Phi}^G(s,n,X)$, and the key that
the users agree on is precisely $K_{\Phi}^G(s,n,X)$. Let us assume now
that $\Phi$ is a transitive action, i.e., for every pair of elements
$s,s'\in S$ there always exists a $g\in G$ such that $\Phi (g,s)=s'$.
Thus if $s\in S$ is a public element, given any two elements in $S$,
$s_1$, $s_2$, there always exist $g_1, g_2\in G$ such that $\phi
(g_i,s)=s_i$, $i=1,2$. Let $s_3=\Phi (g_1, \Phi
(g_2,s))=\Phi (g_1g_2,s)$. If, given $s$, $s_1$ and $s_2$, it is not
feasible to distinguish $s_3$ from a random value in polynomial time,
then an induction argument like that given in \cite[Theorem 1]{steiner2}
allows us to show the following result.

\begin{thm}\label{thmsec}
Let $\Phi$ be a transitive $G$-action on $S$. Then the group key
that users derive as a result of any of the protocols GSAP-1, GSAP-2
and \mbox{GSAP-3} is indistinguishable in polynomial time from a
random value, given only the values exchanged between users during the
protocol,
whenever the corresponding Diffie-Hellman protocol induced by $\Phi$
for two users satisfies this property.
\end{thm}

Another important issue in any group key management is rekeying
after the initial key agreement. There exist three different
situations that require a rekeying operation. The first is
simply due to key caducity and the group of users remains the same. In
the other two cases, we may find a new user that wishes to join the
group or a user who leaves the group. In both situations it is
required that the new (resp.\@ former) user cannot access the former
(resp.\@ new) distributed key. In the following lines we describe the
procedures as well as their security.

\medskip

Let us start by considering the protocol GSAP-1 described in Section
\ref{Sec11}. In this case, we could simply
require that a new initial key agreement is needed. However, we may
shorten the rekeying process, keeping somehow the spirit of the
protocol. If rekeying is due to key caducity, then user ${\cal U}_n$
chooses a new private element $g'_n\in G$ and defines a new sequence
$f_j^n = \Phi (g'_ng_n , C_{j-1} )$, $j = 1,\dots ,n-1$, with
$C_0=s$, as is done in step 3 of GSAP-1. The rest of the users also
proceed as in step 3 and recover
(using their private keys as described in GSAP-1) the new key
$ \Phi \Big(g'_n\prod_{j=1}^{n}g_j , s \Big )$.

In case some user, say ${\cal U}_i$, leaves the group, then
the corresponding value $f_i^n$ is omitted in the new message made
by ${\cal U}_n$.

Finally, in case a user ${\cal U}_{n+1}$ joins the group, then
user ${\cal U}_n$ chooses a new element $g'_n$ and
sends the message
$$\Big \{ \ \Phi (g_n'g_1 , s), \ \Phi (g_n'g_2g_1 , s),\dots , \Phi \Big ( g'_n\prod_{j=1}^{n}g_j  , s \Big)  \ \Big \}$$
to user ${\cal U}_{n+1}$. Then this user starts step 3 of GSAP-1.

Security of all new subsequent key distributions follows from Theorem
\ref{thmsec}.

\medskip

In the case of protocols GSAP-2 and GSAP-3, described in Sections
\ref{Sec12} and \ref{Sec13} respectively, we may use the information
that every user holds after the initial key agreement to rekey very
efficiently as is suggested in \cite{steiner2}. In this case, given
that every user remembers the same information, say
$$\Big\{ \Phi \Big ( \displaystyle{ \prod_{r=2}^{n}} g_r ,s\Big), \Phi \Big ( \displaystyle{ \prod_{r=1; r \neq 2}^{n}} g_r ,s\Big), \dots ,
\Phi \Big ( \displaystyle{ \prod_{r=1; r \neq c}^{n}} g_r,s\Big),
\dots ,\Phi \Big ( \displaystyle{ \prod_{r=1}^{n-1}} g_r ,s\Big)
\Big\},$$

\noindent the rekeying process may be carried out by any one of
them. Let us call this user ${\cal U}_c$. If rekeying is due to key
caducity, then he chooses a new $g'_c\in G$, changes his private
key to $g'_cg_c$ and sends the following rekeying message:

$$\Big\{ \Phi \Big ( g'_c\displaystyle{ \prod_{r=2}^{n}} g_r ,s\Big), \Phi \Big ( g'_c\displaystyle{ \prod_{r=1; r \neq 2}^{n}} g_r ,s\Big), \dots ,
\Phi \Big ( \displaystyle{ \prod_{r=1; r \neq c}^{n}} g_r,s\Big),
\dots ,\Phi \Big ( g'_c\displaystyle{ \prod_{r=1}^{n-1}} g_r ,s\Big)
\Big\}.$$

Then every user, using his private information, recovers the new
common key given by $\Phi \Big ( g'_c\displaystyle{ \prod_{r=1}^{n}}
g_r ,s\Big)$.

In case some user leaves the group, the corresponding position
in the rekeying message is omitted. If a new user joins the group,
then ${\cal U}_c$ adds the element $\Phi \Big ( g'_c\displaystyle{
\prod_{r=1}^{n}} g_r ,s\Big)$ and sends the following to the new user
${\cal U}_{n+1}$:

$$\Big\{ \Phi \Big ( g'_c\displaystyle{ \prod_{r=2}^{n}} g_r ,s\Big),  \dots ,
\Phi \Big ( \displaystyle{ \prod_{r=1; r \neq c}^{n}} g_r,s\Big),
\dots ,\Phi \Big ( g'_c\displaystyle{ \prod_{r=1}^{n-1}} g_r
,s\Big), \Phi \Big ( g'_c\displaystyle{ \prod_{r=1}^{n}} g_r ,s\Big)
\Big\}.$$

\noindent This user proceeds (in both protocols GSAP-2 and
GSAP-3) to step 5 of protocol GSAP-3.

Again, security in every case is a consequence of Theorem
\ref{thmsec}.

\section{Secure group communication based on linear actions}

As can be observed in the protocols given in the previous
section, user ${\cal U}_n$ plays a central role, and in two of them,
every user is required to do a different number of
computations and store a different number of values,
depending on his proximity to ${\cal U}_n$. The aim of this section
is twofold. On one hand, we give a similar approach to that of GSAP-3
in order to get a protocol with the same advantages
that is applicable in situations where the semigroup $G$ acting on $S$
does not contain inverses. On the other hand, we give a new approach
based on linear actions that in some cases not only significantly
decreases communication overheads, but also reduces the number of
rounds to just 2, which will significantly enhance the efficiency.

\medskip

We say that, given $G$ and $S$ semigroups, an action $\Phi:
G\times S \rightarrow S$ defined by $\Phi (g,s)=g\cdot s$ is linear in
case $\Phi(g,ss')=\Phi(g,s)\Phi (g,s')$.

\medskip

The following protocol is a modification of GSAP-3 for a linear
$G$-action $\Phi$ on $S$, but instead of requiring $G$ to be a group,
we require this of $S$. We get a similar
protocol that is also an extension of Diffie-Hellman to the
multiparty case.

\begin{protocol}[GSAP-3']
Users agree on an element $s$ in a finite group $S$, a finite abelian
semigroup $G$, and a linear $G$-action $\Phi$ on $S$. For every $i=1,
\dots ,n$, the user ${\cal U}_i$ holds a private element $g_i\in
G$.

\begin{enumerate}
\item For $i=1, \dots ,n-2$, user ${\cal U}_i$ sends to user  ${\cal
U}_{i+1}$ the message $C_i=\Phi (g_i, C_{i-1})$.

\item User ${\cal U}_{n-1}$ computes $C_{n-1}=\Phi (g_{n-1}, C_{n-2})$
and broadcasts it to the other users
$\{ {\cal U}_1 ,\dots ,{\cal U}_{n-2}, {\cal U}_n \}$.

\item User ${\cal U}_n$ computes the element $\Phi (g_n, C_{n-1})$.

\item  For $i=1, \dots ,n-1$, user ${\cal U}_i$ computes
$D_i=\Phi(g_i, s)^{-1}C_{n-1}$ and sends it to user ${\cal U}_n$.

\item For $i=1, \dots ,n-1$, user ${\cal U}_n$ computes $\Phi (g_n, D_i)$
and sends to users $\{ {\cal U}_1 ,\dots ,{\cal U}_{n-2}, {\cal
U}_{n-1} \}$ the set of values $\{ \Phi (g_n, D_1),\ldots , \Phi
(g_n, D_{n-1}),\linebreak \Phi (g_n, D_n) \}$ and his public key $\Phi
(g_n,s)$, where $D_n=\Phi (g_n,s)^{-1}C_{n-1}$.

\item For $i=1, \dots ,n-1$, user ${\cal U}_i$ computes $\Phi (g_i, \Phi
(g_n,s)) \Phi (g_n, D_i)$.
\end{enumerate}
\end{protocol}

\begin{thm}\label{theoGSAP3'}
After protocol GSAP-3', the users ${\cal U}_1, \dots ,{\cal U}_n$
share a common key given by $\Phi \Big( \displaystyle{\prod_{i=1}^n}
g_i , s\Big)$.
\end{thm}

\begin{proof}
This follows from the linearity of the action
$\Phi$. $\Phi (g_i, \Phi (g_n,s)) \Phi (g_n, D_i)= \Phi (g_ig_n, s)
\Phi \big( g_n, \Phi(g_i, s)^{-1} \Phi \big( \prod_{r=1}^{n-1}g_r,
s\big) \big)= \Phi \big( \prod_{r=1}^{n}g_r, s\big)$, since $\Phi
(g_i,e)=e$, $e$ being the neutral element in $S$, and $\Phi
(g_i,s)^{-1}=\Phi
(g_i,s^{-1})$, again by the
linearity of the action.
\end{proof}

\begin{ex} a) Given again a cyclic group $S$ of order $q$ generated by
$g$, the action $\Phi : \mathbb{N}^*\times S \rightarrow S$
defined by $\Phi (y,g^x)=(g^x)^y$ is clearly linear, so the above
argument applies.  $D_i$ assumes the form
$g^{\prod_{j=1}^{n-1}x_j}g^{-x_i}$.

b) If $\varepsilon$ is the group of points of an elliptic curve, then
$\varepsilon$ is a $\mathbb{Z}$-module via the linear action $\Phi
(k,P)=kP$ for every $k\in \mathbb{Z}$ and $P\in \varepsilon$. $D_i$
assumes the form $(\prod_{i=1}^{n-1}k_j)P-k_iP$.

c) Let us introduce an example where the preceding protocols can be
run over a module structure.
Let us recall from \cite{climent} the following  ring: \[
  E_{p}^{(m)}
  =
  \left\{
    [a_{ij}] \in \mathrm{Mat}_{m \times m} (\mathbb{Z})
    \ | \
    a_{ij} \in \mathbb{Z}_{p^{i}} \ \text{if} \ i \leq j,
    \ \text{and} \
    a_{ij} \in p^{i-j} \mathbb{Z}_{p^{i}} \ \text{if} \ i > j
  \right\},
\]
with addition and multiplication defined, respectively, as
follows
\begin{gather*}
  \big[
    a_{ij}
  \big]
  +
  \big[
    b_{ij}
  \big]
  =
  \big[
    (a_{ij}+b_{ij}) \bmod{p^{i}}
  \big], \\
  \big[
    a_{ij}
  \big]
  \cdot
  \big[
    b_{ij}
  \big]
  =
  \left[
    \left(
      \sum_{k=1}^{m} a_{ik} b_{kj}
    \right)
    \bmod{p^{i}}
  \right].
\end{gather*}
Here $\mathrm{Mat}_{m \times m} (\mathbb{Z})$ denotes the set of $m
\times m$ matrices with entries in $\mathbb{Z}$, and $p^{r}
\mathbb{Z}_{p^{s}}$ denotes the set $\left\{p^{r} u \ | \ u \in \{
0, \ldots, p^s-1\} \right\} \subset \mathbb{Z}$ for positive
integers $r$ and $s$. This ring is clearly non-commutative and its
product defines an action of the multiplicative semigroup
$E_p^{(m)}$ on the set $\mathbb{Z}_{p} \times \mathbb{Z}_{p^{2}}
\times \cdots \times \mathbb{Z}_{p^{m}}$. However, to ensure that
the key exchange works, we need that the elements in the semigroup
commute. In this non-commutative setting, this may be achieved by
considering that the selected elements in the semigroup $E_p^{(m)}$
are of the form $\displaystyle{\sum_{i=0}^{r}}C_iM^i$, such that for
every $i=0,\dots ,r$, $C_i$ is in the center $Z$ of $E_p^{(m)}$ and
$M\in E_p^{(m)}$ is a public element such that its set of powers is
large enough. In other words, if we denote the set of elements of
this form by $Z[M]$, then we are using for $G$ the multiplicative
subsemigroup $Z[M]$ of $E_p^{(m)}$.

From \cite[Theorem 2]{climent2} we can deduce conditions on the
public information that will be sent in order to prevent an
attacker from solving the SAP in the subsemigroup of $Z[M]$ given
by the center $Z$ of the ring, with cardinality
$p^m$ (cf. \cite{climent}). Thus if $M$ has high order, i.e. $M$ is
such that the least integer $n$ satisfying $M^{k+n}=M^k$ for
every sufficiently large $k$ is high, we will obtain that $Z[M]$ is
big enough.

Note that our aim in this paper is not to prove the hardness of the
SAP for this particular example, but rather to present protocols which
rely on the hardness of the SAP in a particular scenario once it has
been established there. The
non-commutative scenario in particular may present hidden
vulnerabilities, as was shown in recent
cryptanalyses, e.g. \cite{kamal,giacomo}, although these seem not to
directly apply in this setting. For example \cite{kamal} introduces
a cryptanalysis for the case of two users when the ring $E_p^{(m)}$
acts on itself, which can be countered by choosing $p$ and $m$
appropriately in order to avoid the existence of inverses
\cite{climent}. In the case of \cite{giacomo}, the cryptanalysis
requires building a system of equations, which does not seem to be
straightforward in this new setting of $Z[M]$. In \cite[Proposition
3.9]{mazeth} it is asserted that if the commutative semigroup has a
big number of invertible elements, then it is possible to develop a
square root attack to the SAP. Again we point out that $E_p^{(m)}$
could be chosen in order to avoid this attack.


 \end{ex}


Given that both $\Phi \Big( \displaystyle{\prod_{i=1}^{n-1}} g_i ,
s\Big)$ and $\Phi (g_n,s)$ are public we immediately get the
following.

\begin{prop}
If an adversary is able to solve the DHSAP, then she can get
the shared key in GSAP-3'.
\end{prop}

Let us recall from \cite{maze} that given any $G$-action $\Phi$ on $S$,
we can easily define an ElGamal type of public key cryptosystem.
We define the following ElGamal type of protocol.

\begin{enumerate}
\item Alice and Bob publicly agree on an element $s\in S$.

\item Bob chooses $b\in G$ and computes $\Phi (b,s)$. Bob's private key is
$b$, his public key is $\Phi (b,s)$.

\item If Alice wants to send the message $m\in S$ to Bob, then she
gets Bob's public key $\Phi (b,s)$.

\item Alice chooses randomly $a\in G$ and computes $\Phi (a,s)$ and $\Phi (a, \Phi
(b,s))$.

\item Alice sends to Bob the pair $(c,d)=\big( \Phi (a,s), m\Phi (a, \Phi
(b,s) \big)$.

\item Bob recovers $m=d\Phi( b,c)^{-1}=m\Phi (a, \Phi
(b,s))\Phi( b, \Phi (a,s))^{-1}$, given that $S$ has a group structure.
\end{enumerate}

It can be easily observed that solving the DHSAP is
equivalent to breaking the preceding algorithm: if given the public
information $$(s, \Phi (a,s), \Phi (b,s), m\Phi (ab,s) )$$ one is able
to get $m$, then the input $(s, \Phi (a,s), \Phi (b,s), e)$, for
$e\in S$ the neutral element, produces $\Phi (ab,s)^{-1}$, which
solves the DHSAP. Conversely, given Bob's public key $\Phi
(b,s)$ and the pair $\big( \Phi (a,s), m\Phi (a, \Phi (b,s))
\big)$, one can use $\Phi (ab,s)$ from the DHSAP to
recover $m$.

\medskip

Now using the above we are able to show the security of GSAP-3'.

\begin{thm}\label{secgsapprime}
The group key that users derive as a result of GSAP-3' is
indistinguishable in polynomial time from a random value whenever
the corresponding Diffie-Hellman protocol induced by $\Phi$ for two
users also satisfies this property.
\end{thm}

\begin{proof}
Given that both $C_{n-1}=\Phi \Big(
\displaystyle{\prod_{i=1}^{n-1}} g_i , s\Big)$ and $D_i=\Phi(g_i,
s)^{-1}C_{n-1}$ are public, an adversary is able to get all the
public values $\Phi (g_i, s)$, $i=1, \dots , n$. Now user ${\cal
U}_n$ sends the message $\{ \Phi (g_n, D_i) \}_{i=1}^{n-1}$ jointly
with $\Phi (g_n, s)$, in other words, due to linearity of $\Phi$,
user ${\cal U}_n$ sends a ``a family of pairs'', $i=1,\ldots,n$,

$$\Big(\Phi (g_n,s), \Phi (g_n, \Phi (g_i,s)^{-1}) \Phi \Big( g_n,
\Phi\Big(
\displaystyle{\prod_{j=1}^{n-1}} g_j , s\Big)\Big)
\Big),$$

\noindent which can be seen as a set of ElGamal encryptions of the message
$$\Phi \Big( \displaystyle{\prod_{i=1}^{n}} g_i , s\Big)=\Phi \Big(
g_n, \Phi\Big( \displaystyle{\prod_{i=1}^{n-1}} g_i , s\Big)\Big)$$

\noindent using the public keys $\Phi (g_i, s)$, $i=1, \dots ,
n$. Alternatively, one can consider the pairs

$$\Big(\Phi (g_i,s), \Phi (g_n, \Phi (g_i,s)^{-1}) \Phi \Big( g_n,
\Phi\Big(
\displaystyle{\prod_{j=1}^{n-1}} g_j , s\Big)\Big)
\Big),$$

\noindent which can also be seen, given the commutativity in $G$, as a set of ElGamal encryptions of the message
$$\Phi \Big( \displaystyle{\prod_{i=1}^{n}} g_i , s\Big)=\Phi \Big(
g_n, \Phi\Big( \displaystyle{\prod_{i=1}^{n-1}} g_i , s\Big)\Big)$$

\noindent using the public key $\Phi (g_n, s^{-1})$, and the $g_i$'s
as random numbers, for $i=1, \dots , n$.

Thus, as we pointed out above, given the equivalence of the security
of the ElGamal type of public key cryptosystem and the DHSAP,
the result follows.
\end{proof}

The rekeying process in this setting is analogous to that described
in Section \ref{Secsec} for protocols GSAP-2 and GSAP-3.

We first note that every user remembers the following keying
information.

$$\{ \Phi (g_n, D_1),\ldots , \Phi (g_n, D_{n-1}), \Phi (g_n, D_n) \}$$

In case of key caducity, user ${\cal U}_c$ for some $c=1, \dots , n$
chooses a new element $g'_c\in G$, computes a new key given by $\Phi
\Big( g'_c \displaystyle{\prod_{i=1}^{n}} g_i , s\Big)$ and his
keying information $\Phi( (g'_c)^2 g_c g_n, s)^{-1}\Phi \Big( g'_c
\displaystyle{\prod_{i=1}^{n}} g_i , s\Big)$ and broadcasts the
following message

\begin{gather*}
\{ \Phi(g'_c, \Phi (g_n, D_1)),\ldots ,\Phi( (g'_c)^2g_cg_n,
s)^{-1}\Phi \Big( g'_c, \Phi\Big(
\displaystyle{\prod_{i=1}^{n}} g_i , s\Big)\Big), \dots ,\\
\Phi(g'_c, \Phi (g_n, D_{n-1})), \Phi(g'_c, \Phi (g_n, D_n)) \},
\end{gather*}

\noindent jointly with the value $\Phi (g'_c,\Phi (g_n,s))$. User
${\cal U}_c$ changes his private information to $g_cg'_c$.

In case rekeying is due to some user leaving the group, then the
corresponding value is omitted in the above message.

Finally, let us assume that ${\cal U}_{n+1}$ joins the group. The
process corresponds in this case to something similar to a ``double
rekeying''as above. First, ${\cal U}_c$ sends to ${\cal U}_{n+1}$

\begin{gather*}
\Big\{ \Phi(g'_c, \Phi (g_n, D_1)),\ldots ,\Phi( (g'_c)^2g_cg_n,
    s)^{-1}\Phi \Big( g'_c, \Phi\Big(
    \displaystyle{\prod_{i=1}^{n}} g_i , s\Big)\Big), \dots ,\\
    \Phi(g'_c, \Phi (g_n, D_{n-1})), \Phi(g'_c, \Phi (g_n, D_n)) , \Phi
    \Big( g'_c, \Phi\Big(
\displaystyle{\prod_{i=1}^{n}} g_i , s\Big)\Big)\Big\}
\end{gather*}

\noindent jointly with the value $\Phi (g'_c,\Phi (g_n,s))$. Then,
${\cal U}_{n+1}$ broadcasts a rekeying message given by

\begin{gather*}
\Big\{ \Phi(g_{n+1}g'_c, \Phi (g_n, D_1)),\ldots ,\Phi(
    g_{n+1}(g'_c)^2g_cg_n, s)^{-1}\Phi \Big( g'_c, \Phi\Big(
    \displaystyle{\prod_{i=1}^{n+1}} g_i , s\Big)\Big), \dots ,\\
    \Phi(g_{n+1}g'_c, \Phi (g_n, D_{n-1})), \Phi(g_{n+1}g'_c, \Phi (g_n,
    D_n)) ,\\
    \Phi(g_{n+1}^2g'_cg_n,s)^{-1} \Phi \Big( g'_c, \Phi\Big(
\displaystyle{\prod_{i=1}^{n+1}} g_i , s\Big)\Big) \Big\}
\end{gather*}

\noindent jointly with the value $\Phi (g_{n+1}g'_cg_n, s)$.

Security of these processes is shown with a similar argument as in
Theorem~\ref{secgsapprime}.

\bigskip

A more symmetrical use of linear actions is the following protocol,
which decreases the number of rounds to just 2, but which
is only applicable in some cases.

\begin{protocol}[GSAP-4]
Users agree on an element $s$ in a finite abelian semigroup $S$, a
finite abelian semigroup $G$, and a linear $G$-action $\Phi$ on $S$.
For every $i=1, \dots ,n$, the user ${\cal U}_i$ holds a private
element $g_i\in G$.

\begin{enumerate}
\item For every $i=1, \dots ,n$, user ${\cal U}_i$ makes public
$\Phi (g_i,s)=g_i\cdot s$.

\item For some $j=1,\dots ,n$, user ${\cal U}_j$ computes and makes
public $$D_i=\Phi \Big( g_j, \prod_{r\not=j,i} \Phi(g_r,s)
\Big),  \ i\not= j, \ i=1, \dots ,n.$$

\item For every $i=1, \dots, n$, $i\not=j$, user ${\cal U}_i$ computes
$D_i\Phi (g_i, \Phi (g_j,s))$. User ${\cal U}_j$ computes $\Phi(
g_j, \big( \prod_{r\not=j} \Phi(g_r,s) \big)$.
\end{enumerate}
\end{protocol}

\begin{thm}\label{theoGSAP4}
After protocol GSAP-4, the users ${\cal U}_1, \dots ,{\cal U}_n$
share a common key given by $\Phi (g_j,\prod_{r\not=j} \Phi(g_r,s))$.
\end{thm}

\begin{proof}
For every $i=1, \dots ,n$, $i\not= j$,

$$\begin{array}{rl} D_i\Phi (g_i, \Phi (g_j,s)) & =\Phi\big( g_j, \prod_{r\not=j,i}
\Phi(g_r,s) \big) \Phi (g_i, \Phi (g_j,s)) \\
& = \Phi\big( g_j, \prod_{r\not=j,i} \Phi(g_r,s) \big) \Phi
(g_ig_j,s) \\
& = \Phi\big( g_j, \prod_{r\not=j,i} \Phi(g_r,s) \big) \Phi
(g_jg_i,s) \\
& = \Phi\big( g_j, \prod_{r\not=j,i} \Phi(g_r,s) \big) \Phi
(g_j,\Phi (g_i,s)) \\
& = \Phi \big(g_j,\prod_{r\not=j} \Phi(g_r,s)\big).
\end{array}$$
\end{proof}

\begin{ex} a) Let us consider again a cyclic group $S$ of order $q$
generated by $g$, with the action $\Phi : \mathbb{N}^* \times S
\rightarrow S$ given by $\Phi (y,g^x)=(g^x)^y$.
Then GSAP-4 implies sharing a key of the form
$K = g^{k_j\sum_{r=1,r\not=j}^nk_r}$.
An adversary can access the messages
$$D_i=\Phi \Big( g_j,  \prod_{r\not=j,i} \Phi(g_r,s) \Big),  \ i\not= j, \ i=1, \dots ,n,$$
from which she can compute $\prod_{r=1,r \ne j}^{n} D_r=K^{n-2}$. In the
case where the order $q$ of $S$ is known, the adversary can now
recover the key $K$ from $K^{n-2}$ by inverting $n-2$ modulo $q$.
This is in particular the case where $S$ is a subgroup of a finite
field, or where it is the group of points of an elliptic curve.
However, we can avoid this weakness by adding some authentication
information as is done in~\cite{ateniese}.

b) Let $m=pq$ with $p$ and $q$ two large primes and let $G =
\mathbb{Z}_{(p-1)(q-1)}^*$. Then the action $\Phi :G\times
\mathbb{Z}_m \rightarrow \mathbb{Z}_m$ given by $\Phi (x,g)=g^x \
\mbox{mod} \ m$ shows an example where the above attack cannot be
developed unless the adversary is able to factorize $m$. The shared
key in this case is of the form $g^{x_j\sum_{i=1,i\not=j}^nx_i} \
\mbox{mod} \ m$.

c) We recall that a semiring $R$ is a semigroup with respect to both
addition and multiplication and the distributive laws hold. It is
also understood that a semiring is commutative with respect to
addition and the existence of neutral elements is not required,
although some authors do require it. Then given a semiring $R$, a left
$R$-semimodule $M$ is an abelian semigroup with an action $\Phi
:R\times M \rightarrow M$, $\Phi(r,m) = rm$, satisfying $r(sm)=(rs)m,
\ (r+s)m=rm+sm$ and $r(m+n)=rm+rn$ for all $r,s\in R$ and $m,n\in M$.
Thus, based on the previous two examples, we can assert in general
that any semimodule $S$ over a semiring $R$ fits with GSAP-4 and the
shared key is of the form $k_j(\sum_{r=1,r\not=j}^nk_r)s$ for $k_i\in
R$, $i=1, \dots ,n$ private and $s\in S$ public.
\end{ex}

\begin{rem}
Due to the attack shown in example a), the hardness of the
Diffie-Hellman problem is not enough to show security in this case.
We leave it as an open question whether the hardness of factoring would be enough to do so.
\end{rem}

\begin{rem}
We can also give protocols based on two-sided actions. To this end
we recall that given a semiring $S$, right $S$-semimodules are
defined dually to left ones. Then, given two semirings $R$ and $S$,
an $(R,S)$-bisemimodule $M$ is both a left $R$-semimodule and a
right $S$-semimodule such that $(rm)s=r(ms)$ for every $r\in R$,
$m\in M$ and $s\in S$.

Now we are able to provide key exchange protocols similar to those
given in the previous sections based on two-sided linear actions
over a $(R,S)$-bisemimodule $M$. 
In the case of
GSAP-3', since we need the existence of inverses with respect to
addition in $M$, we may suppose that $M$ has an $(R,S)$-bimodule
structure for some rings $R$ and $S$.
\end{rem}

\section{Appendix  GSAP1}

\begin{thm} \label{theoGSAP1}
After protocol GSAP-1, users ${\cal U}_1, \dots ,{\cal U}_n$ agree
on the common key $ \Phi \Big(\prod_{j=1}^{n}g_j  , s \Big )$.
\end{thm}

\begin{proof}
    User ${\cal U}_n$ computes $$\Phi (g_n ,
C_{n-1} )= \Phi \Big( g_n,\Phi \Big (\prod_{j=1}^{n-1}g_j , s\Big)
\Big) = \Phi \Big (\prod_{j=1}^{n}g_j ,  s\Big).$$

Let us show now that the rest of the users recover exactly the same key. For
$k=1,\dots ,n-1$, user ${\cal U}_k$ computes $\Phi ( g_k ,
f_{k}^{k+1})$.

It is straightforward to show that for every $i=1, \dots ,n-2$, $j=1, \dots ,n-i-1$, the following equality
holds:

$$f_j^{n-i} = \Phi \Big( \Big ( \prod_{r=n-i}^{n}g_r \Big ) \Big
(\prod_{r=1}^{j-1}g_r \Big ) , s\Big),$$

\noindent with the empty product being equal to 1.

We then have:

$$\begin{array}{rl} f_k^{k+1} & = f_k^{n-(n-k-1)} \\
&= \Phi \Big ( \Big (\prod_{r=k+1}^{n}g_r \Big ) \Big
(\prod_{r=1}^{k-1}g_r \Big ) , s \Big) \\
& = \Phi \Big (\prod_{r=1; r \neq k}^{n}g_r ,  s \Big).
\end{array}$$

Thus, user ${\cal U}_k$ computes $$\Phi ( g_k , f_k^{k+1} )= \Phi
\Big( g_k ,\Phi \Big(\prod_{r=1; r \neq k}^{n}g_r , s \Big) \Big)=
\Phi \Big (\prod_{r=1}^{n}g_r  s\Big),$$

\noindent as we wanted to show.
\end{proof}

\section{Appendix GSAP2}

\begin{thm} \label{theoGSAP2}
After protocol GSAP-2, users ${\cal U}_1, \dots ,{\cal U}_n$ agree
on a common key given by $ \Phi \Big ( \displaystyle{
\prod_{r=1}^{n}}g_r , s\Big )$.
\end{thm}

\begin{proof}
User ${\cal U}_n$ computes $\Phi (g_n
,C_{n-1}^{n-1} )= \smash{\Phi \Big( g_n, \Phi \Big( \displaystyle{
\prod_{r=1}^{n-1}}g_r  , s \Big)\Big)} = \smash{\Phi \Big(
\displaystyle{ \prod_{r=1}^{n}}g_r  , s\Big)}$.

Now, let us show that $f_i^n =  \Phi \Big ( \displaystyle{
\prod_{r=1; r \neq i}^{n}}g_i , s\Big)$ for $i = 1, \ldots, n$.

To do so, we will prove that \ \ $C_s^{i+s} = \smash{\Phi \Big
(\displaystyle{ \prod_{r=1;r \neq i}^{i+s}}g_r  , s\Big)}$ for
$s = 1, \ldots, n-2$ and $i = 1, \ldots, n-s-1$.

Let us make induction on $s$. For $s=1$, we get by definition that
$C_1^{i+1} = \Phi (g_{i+1} , C_{i-1}^{i-1})$. Now it is clear that $C_j^j = \Phi \Big ( \displaystyle{ \prod_{r=1}^{j}}g_r  , s \Big)$
for every $j = 1, \dots ,n-1$. Therefore

$$C_1^{i+1} = \Phi (g_{i+1} , C_{i-1}^{i-1}) = \Phi \Big( g_{i+1},\Phi \Big ( \displaystyle{ \prod_{r=1}^{i-1}}g_r , s \Big) \Big)=
\Phi \Big ( \displaystyle{ \prod_{r=1; r \neq i}^{i+1}}g_r ,
s\Big).$$

Suppose now that $C_{s-1}^{i+s-1} = \Phi \Big (
\displaystyle{\prod_{r=1; r \neq i}^{i+s-1}}g_r , s\Big)$. Then, by
definition,

$$C_s^{i+s} = \Phi ( g_{i+s} , C_{s-1}^{i+s-1} )= \Phi \Big( g_{i+s} ,\Phi \Big ( \displaystyle{ \prod_{r=1; r \neq i}^{i+s-1}}g_r
 , s \Big )\Big) = \Phi \Big ( \displaystyle{ \prod_{r=1; r \neq
i}^{i+s}}g_r  , s\Big).$$

Thus  $C_{n-1-i}^{n-1} = C_{n-1-i}^{i+n-1-i} = \Phi \Big (
\displaystyle {\prod_{r=1; r \neq i}^{i+(n-1-i)}} g_r , s\Big) =
\Phi \Big ( \displaystyle{ \prod_{r=1; r \neq i}^{n-1}}g_r ,
s\Big)$.

Therefore $$f_i^n = \Phi \Big(g_n , \Phi \Big( \displaystyle{
\prod_{r=1; r \neq i}^{n-1}} g_r , s \Big)\Big)=  \Phi \Big (
\displaystyle{ \prod_{r=1; r \neq i}^{n}} g_r , s\Big),$$

\noindent and so user ${\cal U}_i$ computes $\Phi (g_i , f_i^n) =
\smash{\Phi \Big( g_i , \Phi \Big ( \displaystyle{ \prod_{r=1; r \neq i}^{n}}g_r
,s \Big)\Big)} = \smash{\Phi \Big ( \displaystyle{ \prod_{r=1}^{n}}g_r ,
s\Big)}$, as we wanted to show.
\end{proof}

\end{document}